\documentclass[11pt,a4paper]{article}

\usepackage{a4wide}
\usepackage{amsmath,amssymb,amsthm}
\usepackage{eukdate}
\usepackage[pdftex]{graphics}
\usepackage[T1]{fontenc}
\usepackage{lmodern}
\usepackage{mathpartir}
\usepackage[sc]{mathpazo}
\usepackage{stmaryrd}
\usepackage{bm}
\usepackage{url}
\usepackage[all]{xy}
  \CompileMatrices

\theoremstyle{plain}
\newtheorem{theorem}{Theorem}[section]
\newtheorem{corollary}[theorem]{Corollary}
\newtheorem{lemma}[theorem]{Lemma}

\theoremstyle{definition}
\newtheorem{definition}[theorem]{Definition}

\newtheorem{note}[theorem]{Note}

\newtheorem{remark}[theorem]{Remark}
\theoremstyle{remark}

\newcommand{\Abs}[1]{[#1]}
\newcommand{\abs}[1]{\langle#1\rangle}

\newcommand{\act}{\cdot}

\newcommand{\Atom}{\mathbb{A}}

\newcommand{\bij}{\cong}

\newcommand{\cat}[1][C]{\mathbf{#1}}

\newcommand{\comp}{\circ}

\newcommand{\defeq}{\triangleq}

\newcommand{\freshfor}{\mathrel{\#}}

\newcommand{\funfs}[1][]{\mathbin{\shortrightarrow_{\mathrm{fs}}^{#1}}}

\newcommand{\id}{\mathrm{id}}

\newcommand{\imp}{\mathrel{\Rightarrow}}

\newcommand{\morphism}{\rightarrow}

\newcommand{\Nom}[1][]{{\mathbf{Nom}_{#1}}}

\newcommand{\op}{\mathrm{op}}

\newcommand{\Perm}{\textstyle\mathop{\mathrm{Perm}}}

\newcommand{\Powfs}[1][]{\mathrm{P}_{\mathrm{fs}}^{#1}}

\newcommand{\Sch}{\mathbf{Sch}}
\newcommand{\Set}[1][]{\mathbf{Set}_{#1}}

\newcommand{\supp}{\mathop{\mathrm{supp}}}

\newcommand{\swap}[2]{(#1\;#2)}

\newcommand{\TO}{\shortrightarrow}

\newcommand{\Cbox}{\Box}
\newcommand{\FillI}{{\uparrow}}
\newcommand{\FillO}{{\downarrow}}

\newcommand{\OboxI}{\sqcup}
\newcommand{\OboxO}{\sqcap}
\newcommand{\OISub}{\ensuremath{\mathbf{01Sub}}}

\title{An Equivalent Presentation of\\ 
  the Bezem-Coquand-Huber\\
  Category of Cubical Sets}

\author{Andrew Pitts\\
  University of Cambridge Computer Laboratory}

\date{17 September 2013 [updated 21 December 2013]}

\begin{document}

\maketitle

\begin{abstract}
  Staton has shown that there is an equivalence between the category
  of presheaves on (the opposite of) finite sets and partial
  bijections and the category of nominal restriction sets:
  see~\cite[Exercise~9.7]{PittsAM:nomsns}. The aim here is to see that
  this extends to an equivalence between the category of cubical sets
  introduced in \cite{CoquandT:modttc} and a category of nominal sets
  equipped with a `$01$-substitution' operation. It seems to me that
  presenting the topos in question equivalently as $01$-substitution
  sets rather than cubical sets will make it easier (and more elegant)
  to carry out the constructions and calculations needed to build the
  intended univalent model of intentional constructive type theory.
\end{abstract}

\section{Nominal sets}
\label{sec:noms}

I will use notation for nominal sets as in \cite{PittsAM:nomsns}. In
particular:
\begin{itemize}
\item $\Nom$ is the category of nominal sets and equivariant functions
  over a countably infinite set of names $\Atom$.

\item $\Perm\Atom$ is the group of finite permutations of the
  countably infinite set $\Atom$.

\item $\Atom$ also denotes the nominal set of names (permutation
  action: $\pi\act a = \pi\,a$).

\item $2\defeq\{0,1\}$ is the discrete nominal set with two elements
  (trivial permutation action: $\pi\act i = i$).

\item $\supp x$ denotes the smallest finite subset of $\Atom$ that
  supports an element $x$ of a nominal set.

\item $a\freshfor x\;[a\in\Atom, x\in X]$ is the freshness relation
  associated with $X\in\Nom$ (which holds by definition iff
  $a\notin\supp x$).

\item $\Abs{\Atom}X$ is the nominal set of name abstractions
  $\abs{a}x$ of elements $x\in X$ of a nominal set $X$: see
  \cite[chapter~4]{PittsAM:nomsns}. 

\item $\Powfs X$ is the nominal set of finitely supported subsets of a
  nominal set $X$; see \cite[Defintion~2.26]{PittsAM:nomsns}.
\end{itemize}

\section{$01$-Substitution operations}

Let $X$ be a nominal set. A \emph{$01$-substitution operation} on $X$
is a morphism
\[
s\in\Nom(X\times\Atom\times 2, X)
\]
satisfying the following properties, where we write $x(a:=i)$ for
$s(x,a,i)$:
\begin{gather}
  a\freshfor x(a:=i)\label{eq:1} \\
  a\freshfor x \;\imp\; x(a:=i) = x \label{eq:2}\\
  a\freshfor a' \;\imp\; x(a:=i)(a':=i') = x(a':=i')(a:=i)\label{eq:3}
\end{gather}
Note that since $s$ is a morphism is $\Nom$, we also have
\begin{equation}
  \label{eq:4}
  \pi\act(x(a:=i)) = (\pi\act x)(\pi\,a :=i)
\end{equation}
for all $\pi\in\Perm\Atom$. 

\begin{remark}
  Property \eqref{eq:1} tells us that $s$ 
  corresponds to a pair of morphisms in $\Nom(\Abs{\Atom}X,X)$
  \[
  \abs{a}x \mapsto x(a := 0) \quad\text{and}\quad \abs{a}x \mapsto x(a:= 1)
  \]
  and the other two properties imply that these are in fact name
  restriction operations in the sense of
  \cite[section~9.1]{PittsAM:nomsns}.
\end{remark}

\begin{definition}[\textbf{the category of $01$-substitution sets}]
  The category $\OISub$ has objects that are nominal sets equipped
  with a $01$-substitution operation and morphisms $f\in\OISub(X,Y)$ that are
  equivariant functions $f\in\Nom(X,Y)$ preserving the  $01$-substitution
  operation:
  \begin{equation}
    \label{eq:10}
    f(x(a := i)) = (f\,x)(a := i)\,.
  \end{equation}
  Composition and identities are as for ordinary functions. (Note that
  \eqref{eq:10} makes sense from the point of view of a property of
  substitution, only because $f$ is equivariant, which is to say that
  it has empty support as a member of the exponential object $X\funfs
  Y$ in $\Nom$, that is, $(\forall a\in\Atom)\; a\freshfor f \in
  X\funfs Y$.)
\end{definition}

\section{Cubical sets}

Let $\cat$ be the small category whose objects $A$ are finite subsets
of $\Atom$ and whose morphisms $f\in\cat(A,B)$ are functions
$f\in\Set(A,B+2)$ satisfying
\begin{equation}
  \label{eq:5}
  (\forall a,a'\in f^{-1}B)\; f\,a=f\,a' \;\imp a=a'\,.
\end{equation}
The identity morphism $\id_A\in\cat(A,A)$ is the inclusion function
$A\hookrightarrow A+2$:
\begin{equation}
  \label{eq:6}
  (\forall a\in A)\;\id_A\,a = a
\end{equation}
and the composition of $f\in\cat(A,B)$ with $g\in\cat(B,C)$ is $g\comp
f \in\Set(A,C+2)$ given by:
\begin{equation}
  (\forall a\in A)\; (g\comp f)\,a =
  \begin{cases}
    g(f\,a) &\text{if $a\in f^{-1}B$}\\
    f\,a    &\text{if $a\in A-f^{-1}B$.}
  \end{cases}\label{eq:7} 
\end{equation}
The \emph{category of cubical sets} is the category $[\cat,\Set]$ of
presheaves on $\cat^{\op}$. 


\section{$\OISub$ and $[\cat,\Set]$ are equivalent categories}

Let $\cat[I]$ be the subcategory of $\cat$ with the same objects, but
whose morphisms are those $f\in\cat(A,B)$ satisfying $f^{-1}B=A$; in
other words, $\cat[I](A,B)$ consists of all injective functions from
$A$ to $B$. The category $\cat[I]$ has all pullbacks, created by the
inclusion of $\cat[I]$ into $\Set$. The full subcategory of
$[\cat[I],Set]$ consisting of pullback-preserving functors is one
presentation of the \emph{Schanuel topos} and is in particular
equivalent to $\Nom$. Section~6.3 of \cite{PittsAM:nomsns} contains a
detailed account of this equivalence, which I will make use of here.

When we restrict a functor $F\in[\cat,\Set]$ along the inclusion $i:
\cat[I]\morphism \cat$ we get a pullback-preserving functor
$i^*F=F\comp i: \cat[I]\morphism \Set$ because of the following
elementary piece of category theory (this was observed by Staton and
Levy for the category of finite sets and partial bijections, but works
just the same for $\cat$):

\begin{lemma}
  In any category, suppose 
  \begin{equation}
    \label{eq:11}
    \begin{gathered}
      \xymatrix{%
        D \ar[r]^p\ar[d]_q& A \ar[d]^f\\
        B \ar[r]_g & C}    
    \end{gathered}
  \end{equation}
  is a commuting square of monomorphisms for which $p$ and $g$ have
  left inverses $p'$ and $g'$ (so that $p'\comp p = \id_D$ and
  $g'\comp g = \id_B$) making
  \begin{equation}
    \label{eq:13}
    \begin{gathered}
      \xymatrix{%
    D \ar[d]_q& A \ar[l]_{p'}\ar[d]^f\\
    B & C \ar[l]^{g'} } 
    \end{gathered}
  \end{equation}
  commute. Then \eqref{eq:11} is a pullback square. 
\end{lemma}
\begin{proof}
  Exercise! (Hint: use the fact that  $f$ is a monomorphism.)
\end{proof}

To apply this lemma to $\cat$, note that if the morphisms in
\eqref{eq:11} are all in $\cat[I]$, then they have left inverses in
$\cat$: given $f\in\cat[I](A,C)$, we can take
$f'\in\cat(C,A)=\Set(C,A+2)$ to be
\[
f'\,c \defeq
\begin{cases}
  a &\text{if $f\,a =c$ for some (unique) $a\in A$}\\
  0 &\text{if $a\in C-fA$.}
\end{cases}
\]

\begin{corollary}[Staton, Levy]
  \label{cor:1}
  Composing any functor $\cat\morphism\Set$ with the inclusion
  $i: \cat[I]\morphism\cat$ yields a pullback-preserving functor.
\end{corollary}
\begin{proof}
  It is not hard to see that if \eqref{eq:11} is a pullback square in
  $\cat[I]$, then \eqref{eq:13} commutes in $\cat$. So applying any
  functor $\cat\morphism\Set$ to \eqref{eq:11} preserves the
  monomorphisms (because they all have left inverses) and gives a
  square in $\Set$ satisfying the hypotheses of the lemma -- hence
  which is a pullback.
\end{proof}

So we have the following picture:
\begin{equation}
  \label{eq:16}
  \begin{gathered}
    \xymatrix{%
      {I^*:[\cat,\Set]} \ar[dr]^{i^*} \ar@{-->}[r] & {\Sch}
      \ar@{^{(}->}[d] \ar@{}[r]|-{\simeq} & {\Nom}\\ 
      & {[\cat[I],\Set]} & }
  \end{gathered}
\end{equation}
where $\Sch$ is the full subcategory of pullback-preserving functors
and the equivalence $\Sch\simeq \Nom$ is described in
\cite[section~6.3]{PittsAM:nomsns}. From that description of the
equivalence we get the following explicit construction for the functor
$I^*:[\cat,\Set]\morphism \Nom$:

\begin{definition}[\textbf{the functor $I^*:{[\cat,\Set]}\morphism
    \Nom$}] 
  \label{def:1}
  Given $F\in [\cat,\Set]$, the nominal set $I^*F$ consists of
  equivalence classes $[A,x]$ of pairs $(A\in\cat, x\in F\,A)$ for the
  equivalence relation relating $(A,x)$ and $(A',x')$ when there is
  some $B\supseteq A\cup A'$ with $F(A\hookrightarrow B)\,x =
  F(A'\hookrightarrow B)\,x'$. The permutation action on equivalence
  classes is given by $\pi\act[A,x] = [\pi A, F(\pi|_A)\,x]$, where
  $\pi|_A\in\cat[I](A,\pi A)\subseteq \cat(A,\pi A)$ is the injective
  function that $\pi$ gives from the set $A$ to the set $\pi A=
  \{\pi\,a \mid a\in A\}$. It is not hard to see that $A$ supports
  $[A,x]$ with respect to this action, so that $I^*F$ is a nominal
  set.

  Given $\varphi:F\morphism F'$ in $[\cat,\Set]$, $I^*\varphi \in
  \Nom(I^*F,I^*F')$ is the function $I^*\varphi : [A,x] \mapsto
  [A,\varphi_A\,x]$, which is (well-defined and) equivariant because
  $\varphi_A$ is natural in $A$.
\end{definition}

\begin{remark}
  \label{rem:1}
  Since from Corollary~\ref{cor:1} we know that each $F\in
  [\cat,\Set]$ preserves the pullback
  \[
  \xymatrix{%
    {A\cap A'} \ar@{^{(}->}[r] \ar@{^{(}->}[d] & {A'}
    \ar@{^{(}->}[d]\\ 
    A \ar@{^{(}->}[r] & {A\cup A'}
  }
  \]
  the equivalence relation defining $I^*F$ relates $(A,x)$ and
  $(A',x)$ iff  there is some $y\in F(A\cap A')$
  with $F(A\cap A'\hookrightarrow A)\,y = x$ and $F(A\cap
  A'\hookrightarrow A')\,y = x'$.
\end{remark}

We will show that $I^*:[\cat,\Set]\morphism \Nom$ factors through the
forgetful functor $\OISub\morphism\Nom$ to give an equivalence of
categories.

\begin{definition}[\textbf{the $01$-substitution operation on $I^*F$}] 
  \label{def:2}
  Given $F\in[\cat,\Set]$ and $[A,x]\in I^*F$, for each $a\in\Atom$
  and $i\in 2$ we define
  \begin{equation}
    \label{eq:14}
    [A,x](a := i) \defeq [A-\{a\}, F(f_{A,a,i})\,x]
  \end{equation}
  where $f_{A,a,i}\in\cat(A,A-\{a\})$ is the morphism mapping $a$ to
  $i$ if $a\in A$ and otherwise acting like the identity. It is easy
  to see that this definition is independent of the choice of
  representative $(A,x)$. It is equivariant \eqref{eq:4} and satisfies
  property \eqref{eq:3} because the diagrams
  \[
  \xymatrix{%
    A \ar[r]^<<<<<<<{f_{A,a,i}} \ar[d]_{\pi|_A} & {A-\{a\}}
    \ar[d]^{\pi|_{A-\{a\}}} \\
    {\pi A}\ar[r]_<<<<<{f_{\pi A,\pi\,a,i}} & {\pi A- \{\pi\,a\}}
  }
  \text{and}
  \xymatrix{%
    A \ar[r]^<<<<<<<<<{f_{A,a,i}} \ar[d]_{f_{A,a',i'}} & {A-\{a\}}
    \ar[d]^{f_{A-\{a\},a',i'}} \\
    {A-\{a'\}} \ar[r]_>>>>>{f_{A-\{a'\},a,i}} & {A-\{a,a'\}}
    }
    (a\freshfor a')
  \]
  commute in $\cat$. Since $\supp[A,x] \subseteq A$, definition
  \eqref{eq:14} also satisfies property \eqref{eq:1}.  Finally, it
  remains to see that it also satisfies property \eqref{eq:2}. Note that 
  \begin{equation}
    \label{eq:15}
    a\notin A \;\imp\; [A,x](a:=i) = [A,x]
  \end{equation}
  because when $a\notin A$, then $A-\{a\}=A$ and $f_{A,a,i} =
  \id_A$. So if $a\freshfor [A,x]$, then picking any $a'\notin
  A\cup\{a\}$, we have $a'\freshfor[A,x]$ and hence
  $[A,x] = \swap{a}{a'}\act [A,x] = [\swap{a}{a'}A,
  F(\swap{a}{a'}|_A)\,x]$. Now $a\notin\swap{a}{a'}A$, so by
  \eqref{eq:15} $[\swap{a}{a'}A,
  F(\swap{a}{a'}|_A)\,x](a:=i) = [\swap{a}{a'}A,
  F(\swap{a}{a'}|_A)\,x]$; hence $[A,x](a:=i) = [A,x]$, as required
  for  \eqref{eq:2}. 
\end{definition}

Given $\varphi:F\morphism F'$ in $[\cat,\Set]$ and $[A,x]\in I^*F$,
using naturality of $\varphi$ and Definition~\ref{def:1} we can
calculate that
\begin{align*}
  I^*\varphi([A,x](a:=i)) &=
  [A-\{a\},\varphi_{A-\{a\}}(F(f_{A,a,i})\,x)]\\
  &= [A-\{a\}, F(f_{A,a,i})(\varphi_A\,x)]\\
  &= (I^*\varphi[A,x])(a:=i)\,.
\end{align*}
So each $I^*\varphi$ is a morphism in $\OISub$ and $I^*$ lifts to give
a functor $I^*:[\cat,\Set]\morphism \OISub$.

\begin{lemma}
  \label{lem:2}
  $I^*:[\cat,\Set]\morphism \OISub$ is a faithful functor.
\end{lemma}
\begin{proof}
  Since $i:\cat[I]\morphism \cat$ is the identity on objects,
  $i^*:[\cat,\Set]\morphism[\cat[I],\Set]$ is a faithful functor and
  hence so is $I^*:[\cat,\Set]\morphism \Nom$ (\emph{cf.}~diagram
  \eqref{eq:16}). Therefore $I^*$ is faithful as a functor from
  $[\cat,\Set]$ t $\OISub$.
\end{proof}

\begin{lemma}
  \label{lem:3}
  $I^*:[\cat,\Set]\morphism \OISub$ is a full functor.
\end{lemma}
\begin{proof}
  First note that in view of Remark~\ref{rem:1} we
  have for any $F\in[\cat,\Set]$ that 
  \begin{equation}
    \label{eq:17}
    [A,x]=[A,x']\in I^*F \;\imp\; x=x' \in F\,A\,.
  \end{equation}
  Furthermore
  \begin{equation}
    \label{eq:18}
    (\forall d\in I^*F)\; \supp d \subseteq A \;\imp\; (\exists x\in
    F\,A)\; d = [A,x]\,.
  \end{equation}
  For if $\supp [B,y] \subseteq A$, letting $b_1,\ldots,b_n$ be the
  distinct elements of $B-A$, then $b_i\freshfor [B,y]$, so by
  property \eqref{eq:2} for the $01$-substitution set $I^*F$ and
  \eqref{eq:14} we have
  \[
  [B,y] = [B,y](b_1:=0)\cdots(b_n:=0) = [B-\{b_1,\ldots,b_n\}, y'] =
  [A,F(A\cap B\hookrightarrow A)\,y']
  \]
  for some $y'\in F(B-\{b_1,\ldots,b_n\}) =F(A\cap B)$.

  So now suppose $F,F'\in[\cat,\Set]$ and $g\in\OISub(I^*F,
  I^*F')$. For each $A\in\cat$ and $x\in F\,A$, since $g$ is
  equivariant we have $\supp(g[A,x])
  \subseteq \supp[A,x] \subseteq A$. So by \eqref{eq:18}, there is
  some $\varphi_A\,x \in F'A$ with $g[A,x] = [A, \varphi_A\,x]$; and
  by \eqref{eq:17}, $\varphi_A\,x$ is uniquely determined from $A$ and
  $x$ by this property. So we get functions $\varphi_A:F\,A\morphism
  F'A$ for each $A\in\cat$. If we can prove they are natural in $A$,
  then $\varphi\in[\cat,\Set](F,F')$; and $I^*\varphi = g$ by
  construction, as required for fullness. 

  To prove naturality we have to express the $F$ and $F'$ action of an
  arbitrary morphism $f\in\cat(A,B)$ in terms of permutation action
  and $01$-substitution. Note that because of~\eqref{eq:5}, $f$
  restricts to a bijection between $f^{-1}B$ and $f(f^{-1}B)$. Pick a
  finite permutation $\pi\in \Perm\Atom$ that agrees with $f$ on
  $f^{-1}B$ and which is the identity outside the finite set $f^{-1}B
  \cup f(f^{-1}B)$. (We can always find such a $\pi$ -- see the
  Homogeneity Lemma~1.14 in \cite{PittsAM:nomsns}.) Let
  $a_1,\ldots,a_n$ list the distinct elements of $A-f^{-1}B$. Then for
  any $x\in F\,A$
  \[
  [B,(F\,f)\,x] = \pi\act([A,x](a_1:=f\,a_0)\cdots(a_n:= f\,a_n))
  \]
  and similarly for $x'\in F'A$. Therefore since $g$ is a morphism of
  $01$-substitution sets we get:
  \begin{align*}
    [B,\varphi_B((F\,f)\,x)] &= g[B, (F\,f)\,x]\\
    &= g(\pi\act([A,x](a_1:=f\,a_1)\cdots(a_n:= f\,a_n)))\\
    &= (\pi\act (g[A,x]))(a_1:=f\,a_1)\cdots(a_n:= f\,a_n)))\\
    &= (\pi\act [A, \varphi_A\,x]) (a_1:=f\,a_1)\cdots(a_n:=
    f\,a_n)))\\
    &= [B,(F'f)(\varphi_A\,x)]
  \end{align*}
  and hence by \eqref{eq:17} we do indeed have $\varphi_B\comp (F\,f)
  = (F'f)\comp \varphi_A$.
\end{proof}

\begin{theorem}
  \label{thm:1}
  $I^*:[\cat,\Set]\morphism \OISub$ is an equivalence of categories.
\end{theorem}
\begin{proof}
  In view of Lemmas~\ref{lem:2} and \ref{lem:3}, it suffices to check
  that $I^*$ is essentially surjective, that is, for each $X\in
  \OISub$ there is some $I_*X\in [\cat,\Set]$ and an isomorphism
  $\varepsilon_X:I^*(I_*X) \bij X$ in $\OISub$. 

  Given $X\in\OISub$,  for each $A\in\cat$ define 
  \[
  I_*\,X\,A \defeq \{x \in X \mid \supp x \subseteq A\} \in \Set\,.
  \]
  Then for each $f\in\cat(A,B)$, we wish to construct a function
  $I_*\,X\,f\in\Set(I_*\,X\,A,I_*\,X\,B)$. Given $f$, picking $\pi$
  and $a_1,\ldots,a_n$ as in the proof of Lemma~\ref{lem:3}, for each
  $x\in I_*\,X\,A$ we define
  \begin{equation}
    \label{eq:8}
    I_*\,X\,f\,x \defeq \pi\act (x(a_1 := f\,a_1)\cdots(a_n := f\,a_n))\,.
  \end{equation}
  (In the case $n=0$, we take $x(a_1 := f\,a_0)\cdots(a_n := f\,a_n)$ to
  just mean $x$.) Note that since $\supp x\subseteq A$ and using
  \eqref{eq:1}, we have
  \begin{equation}
    \label{eq:9}
    \supp(x(a_1 := f\,a_1)\cdots(a_n := f\,a_n)) \subseteq A -
    \{a_1,\ldots,a_n\} = f^{-1}B
  \end{equation}
  and that by choice of $\pi$, $\pi(f^{-1}B) = f(f^{-1}B) \subseteq
  B$. So the support of the element on the right-hand side of
  \eqref{eq:8} is contained in $B$ and hence it is an element of
  $I_*\,X\,B$. In view of \eqref{eq:9}, the  right-hand side of
  \eqref{eq:8} is independent of the choice of $\pi$; and by
  \eqref{eq:3} it independent of the order in which the elements of
  $A-f^{-1}B$ are listed. So \eqref{eq:8} gives a well-defined function 
  $I_*\,X\,f\in\Set(I_*\,X\,A,I_*\,X\,B)$. One can check that
  $f\mapsto I_*\,X\,f$ preserves identities and composition
  \texttt{[the proof for composition seems very tedious -- I have not
    checked it properly]} and so we get
  $I_*\,X \in [\cat,\Set]$. 

  Note that when $f$ is an inclusion $A\hookrightarrow B$, then in
  \eqref{eq:8} we can take $\pi=\id$ and $n=0$, so that
  \begin{equation}
    \label{eq:19}
    \supp x\subseteq A \;\imp\;I_*\,X(A\hookrightarrow B)\,x = x\,.
  \end{equation}
  If $(A,x)$ and $(A',x')$ both represent the same element of
  $I^*(I_*X)$, then for some $x''\in I_*\,X\,(A\cap A')$ we have
  \[
  x = I_*\,X\,(A\cap A'\hookrightarrow A)\,x'' \quad\text{and}\quad
  x' = I_*\,X\,(A\cap A'\hookrightarrow A')\,x''
  \]
  so that by \eqref{eq:19}, $x=x''=x'$. Therefore we get a
  well-defined function $\varepsilon_X:I^*(I_*\,X) \morphism X$
  satisfying
  \begin{equation}
    \label{eq:21}
    (\forall A\in\cat, x\in I_*\,X\,A)\; \varepsilon_X[A,x] = x\,.
  \end{equation}
 It follows immediately that $\varepsilon_X$ is a bijection. So it
 just remains to check that is it also a morphism in $\OISub$.

 To see that it is equivariant, note that in \eqref{eq:8} when
 $f=\pi|_A$ we have $n=0$ and
  \begin{equation}
    \label{eq:20}
    I_*\,X\,(\pi|_A)\, x = \pi\act x
  \end{equation}
  so that
  \begin{align*}
    \pi\act(\varepsilon_X[A,x]) &= \pi\act x\\
    &= I_*\,X\,(\pi|_A)\, x\\
    &= \varepsilon_X [\pi A, I_*\,X\,(\pi|_A)\, x]\\ 
    &= \varepsilon_X(\pi\act[A,x])\,. 
  \end{align*}
Finally, to see that $\varepsilon_X$
  also preserves the $01$-substitution operation, note that in
  \eqref{eq:8} when $f = f_{A,a,i}$ (Definition~\ref{def:2}), then we
  can take $\pi=\id$, $n=1$ and $a_1=a$ and get
  \begin{equation}
    \label{eq:22}
    I_*\,X\,(f_{A,a,i})\,x = x(a:=i)
  \end{equation}
  and hence
  \begin{align*}
    \varepsilon_X([A,x](a:=i)) &= \varepsilon_X[A-\{a\},
    I_*\,X\,(f_{A,a,i})\,x]\\ 
    &= \varepsilon_X[A-\{a\}, x(a:=i)]\\
    &= x(a:=i)\\
    &= (\varepsilon_X[A,x])(a:=i)\,.
  \end{align*}
\end{proof}

\begin{remark}
  An immediate corollary of the theorem is that $\OISub$ is a
  Grothendieck topos. In fact
  Staton~\cite[section~6.4]{StatonS:namppc} has shown that for a quite
  general notion of `substitution action', categories of nominal sets
  equipped with such actions are all Grothendieck toposes.
\end{remark}

\section{The uniform-Kan condition}

\begin{definition}[\textbf{open boxes}]
  \label{def:opeb}
  Given a non-empty finite subset $A\subseteq_{\mathrm{fin}}\Atom$
  with a distinguished element $a\in A$, an \emph{$1$-open $(A,a)$-box} in
  a $01$-substitution set $X\in\OISub$ is a function
  \[
  u:(A \times 2)-\{(a,1)\} \morphism X
  \]
  satisfying for all $(b,i),(b',i')\in (A\times 2) - \{(a,1)\}$
  \begin{gather}
    b\freshfor u(b,i)\label{eq:28}\\
    u(b,i)(b':=i') = u(b',i')(b:=i)\,.\label{eq:29}
  \end{gather}
  Note that any $x\in X$ gives rise to a $1$-open $(A,a)$-box $u_x$
  with $u_x(b,i) = x(b:=i)$ for all $(b,i)\in (A\times 2) -
  \{(a,1)\}$. We call $x$ a \emph{filling} for the $1$-open
  $(A,a)$-box $u$ if $u=u_x$. Reversing the role of $0$ and $1$ in
  these definitions, we get the notion of \emph{$0$-open $(A,a)$-boxes}
  and their fillings.
\end{definition}

If $c\freshfor A$ and $j\in 2$, then we get another $1$-open
$(A,a)$-box $u(c:=j)$ mapping each $(b,i)\in (A \times 2)-\{(a,1)\}$
to
\begin{equation}
  \label{eq:30}
  (u(c:=j))\,(b,i) = u(b,i)(c:=j)\,.
\end{equation}
Note also that, using the usual permutation action on functions
\begin{equation}
  \label{eq:36}
  \pi\act u = \lambda x \TO \pi\act(u\,(\pi^{-1}\act x))
\end{equation}
if $u$ is a $1$-open $(A,a)$-box, then $\pi\act u$ is a $1$-open
$(\pi\,A,\pi\,a)$-box.

\begin{definition}[\textbf{uniform-Kan objects in $\OISub$}]
  \label{def:unikss}
  A $01$-substitution set $X\in\OISub$ is \emph{uniform-Kan} if it
  comes equipped with operations mapping $1$-open (respectively
  $0$-open) $(A,a)$-boxes $u$ in $X$ for any $(A,a)$, to fillings
  $\FillI u$ (respectively $\FillO u$) in $X$. These operations are
  required to be equivariant
  \begin{equation}
    \label{eq:35}
    \begin{aligned}
      \pi\act\FillI u &= \FillI(\pi\act u) &&\text{if $u$ is
        $1$-open}\\
      \pi\act\FillO u &= \FillO(\pi\act u) &&\text{if $u$ is
        $0$-open}
    \end{aligned}
  \end{equation}
  and to commute with substitution in the sense that if $u$ is an open
  ($A,a)$-box, $c\freshfor A$ and $j\in 2$, then
  \begin{equation}
    \label{eq:32}
    \begin{aligned}
      (\FillI u)(c:=j) &=  \FillI(u(c:=j)) &&\text{if $u$ is
        $1$-open}\\
      (\FillO u)(c:=j) &=  \FillO(u(c:=j)) &&\text{if $u$ is
        $0$-open.}
    \end{aligned} 
  \end{equation}
  \end{definition}

  \begin{note}
    If $X$ is uniform-Kan, the filling operation $\FillI$ gives rise to
    an operation $u\mapsto {}^{+}u$ sending a $1$-open
    $(A,a)$-box $u$ to the $1$-face of its filling $\FillI u$ that is
    orthogonal to the distinguished dimension $a$:
    \begin{equation}
      \label{eq:33}
      {}^{+}u = (\FillI u)(a:=1)\,.
    \end{equation}
    Similarly we get an operation $u\mapsto {}^{-}u$ sending a $0$-open
    $(A,a)$-box $u$ to the $0$-face of its filling:
    \begin{equation}
      \label{eq:61}
      {}^{-}u = (\FillO u)(a:=0)\,.
    \end{equation}
  \end{note}

\begin{note}
  The above definition of uniform-Kan can be reformulated in a less
  `nominal' fashion by making use of name
  abstraction~\cite[chapter~4]{PittsAM:nomsns}, as follows (but it is
  not clear that is any more useful when formulated that way):

  Given $X\in\OISub$, the $01$-substitution operation on $X$ lifts to
  the nominal set $\Abs{\Atom}X$ of name
  abstractions, where it satisfies
  \begin{equation}
    \label{eq:23}
    a\freshfor a' \;\imp\; (\abs{a}x)(a':=i) = \abs{a}(x(a':=i)) 
  \end{equation}
  and hence gives a $01$-substitution operation for $\Abs{\Atom}X$.
  Denote the resulting object of $\OISub$ by $\Cbox X$. Iterating this
  construction, define
  \begin{equation}
    \label{eq:24}
    \begin{cases}
      \Cbox_0X &= X\\
      \Cbox_{n+1}X &= \Cbox(\Cbox_nX)\,.
    \end{cases}
  \end{equation}
  Thus $\Cbox_nX$ is the nominal set of $n$-ary name abstractions
  $\abs{a_1,\dots,a_n}x$ (with $a_1,\ldots,a_n$ mutually distinct),
  with $01$-substitution operation satisfying the evident
  generalization of \eqref{eq:23} to $n$-ary name abstractions.
  
  We can think of the elements of $\Cbox X$ as \emph{intervals} in $X$: given
  $\abs{a}x\in\Cbox X$ its endpoints in $X$ are $\delta^1_0\abs{a}x = x(a:=0)$
  and $\delta^1_1\abs{a}x = x(a:=0)$. Note that the functions
  $\delta^1_0$ and $\delta^1_1$ are morphisms in $\OISub$ from $\Cbox X$
  to $X$. In general at higher dimensions, think of the elements of
  $\Cbox_nX$ as \emph{$n$-cubes} in $X$. There are face morphisms
  \begin{equation}
    \label{eq:25}
    \begin{gathered}
      \delta^m_i \in \OISub(\Cbox_n X, \Cbox_{n-1}X) \qquad(1\leq m\leq n,
      i=0,1)\\
      \delta^m_i\abs{a_1,\ldots,a_n}x =
      \abs{a_1,\ldots,a_{m-1},a_{m+1},\ldots,a_n}x(a_m:=i) 
    \end{gathered}
  \end{equation}
  and degeneracy morphisms
  \begin{equation}
    \label{eq:26}
    \begin{gathered}
      \iota^m\in\OISub(\Cbox_n X,\Cbox_{n+1}X) \qquad (0\leq m\leq n)\\
      \iota^m \abs{a_1,\ldots,a_n}x =
      \begin{array}[t]{@{}l}
        \abs{a_1,\ldots,a_{m-1},a,a_{m+1},\ldots,a_n}x\\
        \quad\text{for some/any $a\freshfor(a_1,\ldots,a_n,x)$.}
      \end{array}
    \end{gathered}
  \end{equation}
  We get a $01$-substitution set $\OboxI_nX$ of $1$-open boxes in $X$
  of dimension $n+1$: it elements are $(n+1)$-ary name abstractions
  $\abs{a\vec{a}}u$ where $a\vec{a}$ are $n+1$ distinct names and $u$
  is a $1$-open $(\{a\vec{a}\},a)$-box. Recalling from
  Definition~\ref{def:opeb} that each $x\in X$ gives a $1$-open
  $(\{a\vec{a}\},a)$-box $u_x$, we get a morphism in $\OISub$:
  \begin{equation}
    \label{eq:34}
    \begin{gathered}
      p_n\in \OISub(\Cbox_{n+1}X,\OboxI_nX)\\
      p_n(\abs{a\vec{a}}x) = \abs{a\vec{a}}u_x
    \end{gathered}
  \end{equation}
  Symmetrically, there is a $01$-substitution set $\OboxO_n X$ of open
  $0$-boxes in $X$ of dimension $n+1$, together with a morphism
  $q_n\in \OISub(\Cbox_{n+1}X,\OboxO_nX)$.

  Then $X$ is uniform-Kan iff each $p_n$ and each $q_n$ is split, that
  is, there are morphisms
  \[
  i_n\in\OISub(\OboxI_nX,\Cbox_{n+1}X) \quad\text{and}\quad
  j_n\in\OISub(\OboxO_nX,\Cbox_{n+1}X)
  \]
  with $p\comp i = \id_{\OboxI_nX}$ and $q\comp j = \id_{\OboxO_nX}$.
\end{note}

Definitions~\ref{def:opeb} and \ref{def:unikss} generalize to
indexed families in $\OISub$ as follows:
\begin{definition}[\textbf{uniform-Kan fibrations}]
  \label{def:unikf}
  Given $p:X\morphism Y$ in $\OISub$, a $1$-open $(A,a)$-box in $X$
  \emph{lies over} $y\in Y$ if $p(u(b,i)) = y(b:=i)$ for all $(b,i)\in
  (A\times 2)-\{(a,1)\}$. Such a $u$ has a \emph{filling over $y$} if
  it has a filling $x\in X$ with $p\,x=y$. Note that if $u$ is a
  $1$-open $(A,a)$-box over $y$, then $\pi\act u$ is a $1$-open
  $(\pi\,A,\pi\,a)$-box over $\pi\act y$; and if $c\freshfor(A,y)$ and
  $j\in 2$, then $u(c:=j)$, defined as in \eqref{eq:30}, is also
  $1$-open $(A,a)$-box over $y$ (since $y(c:=j)= y$).

  Then $p:X\morphism Y$ is a \emph{uniform-Kan fibration} if for each
  $y\in Y$ there are operations mapping any $1$-open (respectively,
  $0$-open) box $u$ over $y$ to an element $\FillI u$ (respectively
  $\FillO u$) in $X$ that is filling over $y$ for $u$; furthermore,
  the operations $\FillI$ and $\FillO$ are required to be
  equivariant~\eqref{eq:35} and commute with
  substitutions~\eqref{eq:32} for names $c$ satisfying not only
  $c\freshfor A$, but also $c\freshfor y$.
\end{definition}


\end{document}